\documentclass[letterpaper, 10pt, conference]{IEEEconf}      % Use this line for a4
\IEEEoverridecommandlockouts                              % This command is only
\usepackage{graphics} % for pdf, bitmapped graphics files
\usepackage{epsfig} % for postscript graphics files
\usepackage{mathptmx} % assumes new font selection scheme installed
\usepackage{times} % assumes new font selection scheme installed
\usepackage{amsmath} % assumes amsmath package installed
\usepackage{amssymb}  % assumes amsmath package installed
\usepackage{eucal}

\usepackage{amsthm}

\newtheorem{thm}{Theorem}
\newtheorem{claim}{Claim}
\newtheorem{corollary}{Corollary}

\newtheorem{rem}[thm]{Remark}

\newtheorem{definition}{Definition}
\newtheorem{definitionn}{}

\newtheorem{assumption}{}

\newcommand{\newdet}{\,\mathrm{det}}% Define dif

\title{\LARGE \bf
A J-Spectral Factorization Condition for the Physical Realiazability of a Transfer Function Matrix with only Direct Feedthrough Quantum Noise}

\author{Rebbecca TY Thien, Shanon L. Vuglar and Ian R. Petersen% <-this % stops a space
\thanks{This work was supported by the Australian Research Council under grant DP210101938. It was also supported by the Office of Naval Research Global under agreement number N62909-19-2129.}% <-this % stops a space
\thanks{Rebbecca TY Thien is with the School of Engineering, The Australian National University, Canberra ACT 2601, Australia
        {\tt\small rebbecca.thien@anu.edu.au}}%
\thanks{Shanon L. Vuglar is with the College of STEM and Health Professions, John Brown University, Siloam Springs AR 72761, United States of America
        {\tt\small shanonvuglar@gmail.com}}%
\thanks{Ian R. Petersen is with the School of Engineering, The Australian National University, Canberra ACT 2601, Australia 
        {\tt\small i.r.petersen@gmail.com}}%
}

\begin{document}
\maketitle
\thispagestyle{empty}
\pagestyle{empty}

%%%%%%%%%%%%%%%%%%%%%%%%%%%%%%%%%%%%%%%%%%%%%%%%%%%%%%%%%%%%%%%%%%%%%%%%%%%%%%%%
\begin{abstract}

% Physically realizable quantum linear system is 
% In general, quantum noise can be detrimental to the performance of a coherent quantum feedback control system. 
%what strictly proper transfer function?
This paper gives a J-spectral factorization condition for the implementation of a strictly proper transfer function matrix as a physically realizable quantum system using only direct feedthrough quantum noise. A necessary frequency response condition is also presented. Examples are included to illustrate the main results. 

\end{abstract}

%%%%%%%%%%%%%%%%%%%%%%%%%%%%%%%%%%%%%%%%%%%%%%%%%%%%%%%%%%%%%%%%%%%%%%%%%%%%%%%%
\section{INTRODUCTION}

Quantum linear systems are a class of quantum systems whose dynamics take the specific form of a set of linear quantum stochastic differential equations~(QSDEs). Such linear quantum systems are common in the area of quantum optics~\cite{BachorRalph, WallsMilburn, GardinerZoller}.
Generally, a set of linear QSDEs need not correspond to a physically meaningful quantum system as they must satisfy additional constraints to represent a physical quantum system. 
The laws of quantum mechanics dictate that quantum systems evolve unitarily, implying that~(in the Heisenberg picture) certain canonical commutation relations~(CCR) are satisfied at all times. 
The notion of a physically realizable quantum linear stochastic system can be seen in~\cite{JamesNurdinPetersen} where the authors also derive a necessary and sufficient characterization for such systems. 

The authors in~\cite{VuglarPetersen} provided a condition given in terms of a non-standard algebraic Riccati equation for physically realizing a given transfer function matrix by only introducing direct feedthrough quantum noises. 
It is well-known~\cite{Molinari} that there is a relationship between spectral factorization and the solution of an algebraic Riccati equation. 
This motivates us to consider a J-spectral factorization approach to physically realize a given transfer function matrix.

In this work, we present a condition for realizing a given transfer function matrix in terms of a J-spectral factorization problem~\cite{Stefanovski,Kimura}.
This also leads to a necessary frequency response condition. 

The remainder of the paper proceeds as follows. In Section~\ref{backandrel}, we describe the quantum linear system models under consideration and define the corresponding notion of physical realizability. This section also gives some preliminary results. Then, in Section~\ref{mainresult}, we present our main results. Examples are given in Section~\ref{example} followed by a conclusion and future work in Section~\ref{conclusion}.
%%%%%%%%%%%%%%%%%%%%%%%%%%%%%%%%%%%%%%%%%%%%%%%%%%%%%%%%%%%%%%%%%%%%%%%%%%%%%%%%
\section{BACKGROUND AND PRELIMINARY RESULTS} \label{backandrel}
\subsection{Quantum Linear Systems} 

The linear quantum systems considered here can be described by the following linear quantum stochastic differential equations~(LQSDEs)~\cite{HudsonParthasarathy, Parthasarathy, Belavkin, JamesNurdinPetersen,VuglarPetersen}:
    \begin{equation} \label{classicmodel}
        \begin{split}
            dx(t) & = Ax(t)dt+Bdw(t); \\
            dy(t) & = Cx(t)dt+Ddw(t) \\
        \end{split}
    \end{equation}
where $A, B, C$ and $D$ are real matrices in $\mathbb{R}^{n \times n}$, $\mathbb{R}^{n \times n_w}$, $\mathbb{R}^{n_y \times n}$ and $\mathbb{R}^{n_y \times n_w}$~($n, n_w, n_y$ are even positive integers), respectively. Moreover, $x(t)=[x_1(t) ... x_n(t)]$ is a column vector of self-adjoint, possibly non-commutative, system variables. 
 
Equations~(\ref{classicmodel}) must also preserve certain \textit{commutation relations} as follows:
    \begin{equation} \label{ccr}
        [x_j(t), x_k(t)]=x_j(t)x_k(t)-x_k(t)x_j(t)=2i \Theta_{jk}
    \end{equation}
where $\Theta$ is a real skew-symmetric matrix with components $\Theta_{jk}$ where $j,k = 1, ..., n$ and $i=\sqrt{-1}$ in order to represent the dynamics of a physically meaningful quantum system.

The \textit{commutation relations} (\ref{ccr}) are said to be \textit{canonical} (i.e., the system is fully quantum) if 
    \begin{equation} \label{theta}
        \Theta_m = diag(J, J, ..., J)
    \end{equation}
    where $J$ denotes the real skew-symmetric $2 \times 2$ matrix
    \begin{equation*}
        J =\begin{bmatrix} 
        0 & 1 \\
        -1 & 0
        \end{bmatrix}
    \end{equation*} 
and the ``\textit{diag}'' notation indicates a block diagonal matrix assembled from the given entries. Here $m$ denotes the dimension of the matrix $\Theta_m$.

The vector quantity $w$ describes the input signals and is assumed to admit the decomposition 
    \begin{equation*}
        dw(t)=\beta_w(t)dt+d\Tilde{w}(t)
    \end{equation*}
where the self-adjoint, adapted process $\beta_w(t)$ is the signal part of $dw(t)$ and $d\Tilde{w}$ is the noise part of $dw(t)$ \cite{HudsonParthasarathy, Parthasarathy,Belavkin}.The noise $\Tilde{w}(t)$ is a vector of self-adjoint quantum noises with Ito table
    \begin{equation*}
        d\Tilde{w}(t)d\Tilde{w}^T(t)=F_{\Tilde{w}}dt
    \end{equation*}
where $F_{\Tilde{w}}=S_{\Tilde{w}}+T_{\Tilde{w}}$ is a nonnegative Hermitian matrix \cite{Belavkin,Parthasarathy} with $S_{\Tilde{w}}$ and $T_{\Tilde{w}}$ are real and imaginary, respectively. 
In this paper, we will assume $F_{\Tilde{w}}$ is of the form $F_{\Tilde{w}} = I + i\Theta$ where $\Theta$ is of the form~(\ref{theta}).

In this work, we consider a special case of~(\ref{classicmodel}):
    \begin{equation} \label{specialcase}
        \begin{split}
            dx(t) & = Ax(t)dt+B_udu(t)+B_{v}dv(t); \\
            dy(t) & = Cx(t)dt+dv(t); \\
        \end{split}
    \end{equation}
see also \cite{JamesNurdinPetersen, VuglarPetersen}.
Here, $dw(t)$ from~(\ref{classicmodel}) has been partitioned into the signal input, $du(t)$~(a column vector with $n_u$ components) and the direct feed through quantum vacuum noise input, $dv(t)$. 
We could regard such a quantum system as a coherent controller in a coherent quantum feedback control system; e.g., see~\cite{JamesNurdinPetersen, VuglarPetersen}.

%%%%%%%%%%%%%%%%%%%%%%%%%%%%%%%%%%%%%%%%%%%%%%%%%%%%%%%%%%%%%%%%%%%%%%%%%%%%%%%%

\subsection{Physical Realizability}
In~\cite{JamesNurdinPetersen}, the notion of physical realizability based around the concept of an open quantum harmonic oscillator is introduced. The following formally defines physical realizability.
    \begin{definition}
    The system~(\ref{classicmodel}) is said to be physically realizable if $\Theta$ is canonical and there exists a quadratic Hamiltonian operator \begin{math}\mathcal{H}=(1/2)x(0)^TRx(0) \end{math}, where $R$ is a real symmetric \begin{math}n \times n \end{math} matrix, and a coupling operator \begin{math} \mathcal{L} = \Lambda x(0) \end{math}, where $\Lambda$ is a complex-valued $\frac{n_w}{2} \times n$ coupling matrix such that matrices $A, B, C,$ and $D$ are given by 
        \begin{subequations} \label{oqhs}
            \begin{align}
                A & = 2 \Theta (R + \Im{(\Lambda^\dagger \Lambda)}) \label{A}\\
                B & = 2i\Theta[-\Lambda^\dagger \quad \Lambda^T]\Gamma \label{B}\\ 
                C & = P^T \label{C}
                    \begin{bmatrix} 
                        \Sigma_{n_y} & 0 \\
                        0 & \Sigma_{n_y} 
                    \end{bmatrix}
                    \begin{bmatrix} 
                    \Lambda+\Lambda^\# \\
                    -i\Lambda+i\Lambda^\# \end{bmatrix} \\
                    D & = [I_{n_y \times n_y} \quad 0_{n_y \times (n_w-n_y)}].\label{D}
            \end{align}
        \end{subequations}
    Here 
        \begin{equation*}
            \begin{split}
                \Gamma & =P_{N_w}diag_{N_w}(M); \\
                M & =\frac{1}{2}    
                    \begin{bmatrix} 
                        1 & i \\
                        1 & -i 
                    \end{bmatrix}; \\
                \Sigma_{N_y} & =[I_{N_y \times N_y} \quad 0_{N_y \times(N_w-N_y)}]; \\
                P_{N_w}(a_1, a_2,..., a_{2N_w})^T & =(a_1,..., a_{2N_w-1}, a_2,..., a_{2N_w})^T;
            \end{split}
        \end{equation*} 
    and diag$(M)$ is an appropriately dimensioned square block diagonal matrix with each diagonal block equal to the matrix $M$. 
    Note that the permutation matrix $P$ has the unitary property $PP^T=P^TP=I$ and $N_w=n_w/2$ and $N_y=n_y/2$.
    \end{definition}
    
The following theorem~\cite{JamesNurdinPetersen} gives necessary and sufficient conditions for the physical realizability of our system~(\ref{specialcase}).
    \begin{thm} \label{theoremPR}~\cite[Theorem 3.4]{JamesNurdinPetersen}
        The system~(\ref{specialcase}) is physically realizable if and only if
            \begin{equation*} 
                \begin{split}
                    A\Theta_n + \Theta_n A^T + B_{v}\Theta_{n_v} B_{v}^T + B_u\Theta_{n_u} B_u^T & = 0; \\
                    B_{v} \begin{bmatrix} 
                    I_{n_y \times n_y} \\
                    0_{(n_w-n_y)\times n_y}
                    \end{bmatrix} & = \Theta C^T diag(J); 
                \end{split}
            \end{equation*}
        where $\Theta_n, \Theta_{n_v}$ and $\Theta_{n_u}$ are all defined as in~(\ref{theta}) but may be of different dimensions.
\end{thm}
Here $(.)^{\dagger}$ denotes the complex conjugate transpose of a matrix while $(.)^{\#}$ denotes the complex conjugate of a matrix.

We consider a strictly proper $n_y \times n_u$ transfer function matrix $G(s)$ with McMillan degree~\cite{Kalman} $n$ where $n$, $n_u$ and $n_y$ are all even.

\begin{definition}\label{dfn2}
Consider an $n_y \times n_u$ strictly proper transfer function matrix $G(s)$ of McMillan degree $n$ where $n$, $n_u$ and $n_y$ are all even. $G(s)$ is said to be physically realizable with only direct feedthrough quantum noise if there exists a minimal realization $G(s)=C(sI-A)^{-1}B_u$ and a matrix $B_{v}$ such that the system~(\ref{specialcase}) is physically realizable.
\end{definition}
The following theorem from \cite{VuglarPetersen} gives a state-space condition in the form of a non-standard algebraic Riccati equation~(NSARE) under which a strictly proper transfer function can be implemented as a physically realizable quantum system, which only introduces direct feedthrough quantum noise.

\begin{thm}\label{thm2}
Consider an $n_y \times n_u$ strictly proper transfer function matrix $G(s)$ of McMillan degree $n$ with minimal state-space realization
    \begin{equation}\label{G(s)}
        G(s)=C(sI-A)^{-1}B_{u}
    \end{equation}
where $n$, $n_u$ and $n_y$ are all even.
This transfer function matrix is physically realizable with only direct feedthrough quantum noise if and only if the algebraic Riccati equation
    \begin{equation} \label{NSARE}
        A^TX + XA - XB_u\Theta_{n_u}B_u^TX + C^T\Theta_{n_y}C = 0
    \end{equation}
has a non-singular, real, skew-symmetric solution X. 
\end{thm}
\begin{proof}
Sufficiency follows directly from Theorem~2 of~\cite{VuglarPetersen}. Necessity is straighftorward to verify from Theorem~1 of~\cite{VuglarPetersen} using a suitable state-space transformation.
\end{proof}

Associated with the state-space realization~(\ref{G(s)}) and the Riccati equation~(\ref{NSARE}) is the Hamiltonian matrix
\begin{equation}\label{hamil}
    H  = 
    \begin{bmatrix}
    A & -B_u\Theta_{n_u}B_u^T \\
    -C^T\Theta_{n_y}C & -A^T
    \end{bmatrix}
\end{equation}
where $\Theta_{n_u}$ and $\Theta_{n_y}$ are defined as in~(\ref{theta}). 

\begin{rem}\label{remark}
Note that $H$ and $-H^T$ are similar whereby $\lambda_H$ is an eigenvalue of $H$ if and only if $-\lambda_H^{\dagger}$ is also an eigenvalue of $H$; e.g., see~\cite[pp. 327-328]{ZhouDoyleGlover}.
\end{rem}

In this paper, instead of a state-space condition, we want to find a frequency response condition such that a given transfer function matrix is a physically realizable quantum system with only direct feedthrough quantum noise.
To achieve this, we use a characterization of the existence of a solution to the NSARE~(\ref{NSARE}) in terms of the J-spectral factorization of a rational matrix following the approach of~\cite{Molinari}.
%%%%%%%%%%%%%%%%%%%%%%%%%%%%%%%%%%%%%%%%%%%%%%%%%%%%%%%%%%%%%%%%%%%%%%%%%%%%%%%%%%%%%%%%%%%%%%%%%

\section{Main Result} \label{mainresult}
In this section, we will show that a given transfer function matrix $G(s)$ is physically realizable with only direct feedthrough noise if and only if there exists a J-spectral factorization of a certain transfer function matrix $\Phi_J(s)$. This $n_u \times n_u$ matrix $\Phi_J(s)$ is defined as
\begin{equation}\label{oriphij}
    \Phi_J(s) = \Theta_{n_u} - G^\sim(s) \Theta_{n_y}G(s)
    %\Phi_J(s) = \Theta_{n_u} - B_u^T(-sI-A^T)^{-1}C^T\Theta_{n_u}C~(sI-A)^{-1}B_u
\end{equation}
where $G^\sim(s) = G(-s)^T$. 
We now consider some assumptions of $G(s)$. 
For a given minimal realization $G(s) = C(sI-A)^{-1}B_u$, we assume the following
    \begin{assumption}\label{assumption1}
        The matrix $A$ is Hurwitz;
    \end{assumption}
    \begin{assumption}\label{assumption3}
        The matrix $A$ and the Hamiltonian matrix defined in~(\ref{hamil}) $H$ have no common eigenvalues.
    \end{assumption}
Note, it follows from the property of Hamiltonian matrices given in Remark~\ref{remark} that Assumption~\ref{assumption3} also implies that the matrix $-A^T$ and the matrix $H$ will have no common eigenvalues.
    
% \begin{definition}\label{definition}
% The $2n \times 2n$ matrix $\Phi_J(s)$ is defined as such:
% \begin{equation}\label{oriphij}
%     \Phi_J(s) = \Theta_{n_u} - B_u^T(-sI-A^T)^{-1}C^T\Theta_{n_u}C~(sI-A)^{-1}B_u.
% \end{equation}
% \end{definition}

\subsection{A J-Spectral Factorization Problem}
The J-spectral factorization problem considered in this paper is defined as follows:
\begin{definition}\label{definition1}
An $n_u \times n_u$ rational matrix $N(s)$ defines a J-spectral factorization of $\Phi_J(s)$ if the following conditions hold:
    \begin{definitionn}
    $\Phi_J(s) = N^{\sim}(s)\Theta_{n_u}N(s)$;\label{d1}
    \end{definitionn}
     \begin{definitionn}
    $N(s)$ is analytic in Re $s \geq 0$; \label{d2}
    \end{definitionn}
     \begin{definitionn}
    $N^{-1}(s)$ has no poles in common with $N(s)$; \label{d3}
    \end{definitionn}
     \begin{definitionn}
    $\lim_{s\to\infty} N(s) = I$. \label{d4}
    \end{definitionn}
\end{definition}

\begin{thm}\label{theorem}
Let $G(s)$ be a given $n_y \times n_u$ strictly proper transfer function matrix with McMillan degree $n$ and minimal realization $G(s) = C(sI-A)^{-1}B_u$ where $n$, $n_u$ and $n_y$ are all even. 
Also, let $\Phi_J(s)$ be defined as in~(\ref{oriphij}) and suppose Assumptions~\ref{assumption1}-~\ref{assumption3} are satisfied. 
Then $G(s)$ is physically realizable with only direct feedthrough quantum noise if and only if $\Phi_J(s)$ has a J-spectral factorization.
\end{thm}

\begin{proof}
    The proof is structured as follows and follows~\cite{Molinari}: We prove necessity and sufficiency for the existence of a skew-symmetric solution $X$ to the NSARE~(\ref{NSARE}) and then apply Theorem~\ref{thm2}. 
    
Necessity: Suppose $G(s)$ is physically realizable with only direct feedthrough quantum noise. 
It follows from Theorem~\ref{thm2} that there exists an X which is a skew-symmetric solution of the NSARE~(\ref{NSARE}). 
Let $N(s)$ be defined as follows
\begin{equation}\label{N}
    N(s) = I + \Theta_{n_u}B_u^TX(sI-A)^{-1}B_u.
\end{equation}
We first show $N(s)$ satisfies~\ref{d1} in Definition~\ref{definition1}. Indeed
    \begin{equation}\label{theoremeqn1}
        \begin{split}
        & N^{\sim}(s)\Theta_{n_u}N(s) \\
        & = [I+B_u^T(-sI-A^T)^{-1}XB_u\Theta_{n_u}]\Theta_{n_u}\\
        &\qquad \qquad \qquad \qquad \times [I+\Theta_{n_u}B_u^TX(sI-A)^{-1}B_u]\\
        & = \Theta_{n_u} - B_u^TX(sI-A)^{-1}B_u - B_u^T(-sI-A^T)^{-1}XB_u \\
        & - B_u^T(-sI-A^T)^{-1}XB_u\Theta_{n_u}B_u^TX(sI-A)^{-1}B_u.
        \end{split}
    \end{equation}
Also, the NSARE~(\ref{NSARE}) implies
    \begin{equation*}
        -(-sI-A^T)X-X(sI-A)+C^T\Theta_{n_y}C=XB_u\Theta_{n_u}B_u^TX
    \end{equation*}
    for any $s \in \mathbb{C}$ and hence
    \begin{multline*}
    -X(sI-A)^{-1}-(-sI-A^T)^{-1}X \\ 
    +(-sI-A^T)^{-1}C^T\Theta_{n_y}C(sI-A)^{-1} \\
    =(-sI-A^T)^{-1}XB_u\Theta_{n_u}B_u^TX(sI-A)^{-1}.
    \end{multline*}
Substituting this result into equation~(\ref{theoremeqn1}), it follows that
    \begin{equation*}
        \begin{split}
        & N^{\sim}(s)\Theta_{n_u}N(s) \\
        & = \Theta_{n_u} - B_u^T(-sI-A^T)^{-1}C^T\Theta_{n_u}C~(sI-A)^{-1}B_u \\
        & = \Phi_J(s).
        \end{split}
    \end{equation*}
Thus, we have established~\ref{d1} of Definition~\ref{definition1}.

In order to establish~\ref{d2} of the definition, note that Assumption~\ref{assumption1} implies that the $N(s)$ is analytic in $Re$ $s \geq 0$.

To show condition~\ref{d3}, note that $\Phi_J(s)$ in~(\ref{oriphij}) can be rewritten in the form
\begin{equation*}
    \Phi_J(s) = \Theta_{n_u} + 
        \begin{bmatrix} 
        0 & B^T_u
        \end{bmatrix}
        \begin{bmatrix} 
        sI-A & 0 \\
        -C^T\Theta_{n_y}C & sI + A^T
        \end{bmatrix}^{-1}
        \begin{bmatrix} 
        B_u \\ 0
        \end{bmatrix}.
\end{equation*} 
Taking the determinant of the above equation and using the determinant relation from~\cite[p. 135, Fact 2.14.13]{Bernstein}, it follows that
\begin{equation*}
\begin{split}
       & \newdet\left(\begin{bmatrix} 
        sI-A & 0 \\
        -C^T\Theta_{n_y}C & sI + A^T
        \end{bmatrix}\right)\newdet(\Phi_J(s)) \\ & = \mathrm{det}(\Theta_{n_u})\\
        & \qquad \times \newdet\left(\begin{bmatrix} 
        sI-A & 0 \\
        -C^T\Theta_{n_y}C & sI + A^T
        \end{bmatrix}
        + \begin{bmatrix} 
        B_u \\
        0
        \end{bmatrix}\Theta^{-1}_{n_u}\begin{bmatrix} 
        0 & B^T_u
        \end{bmatrix}\right)
\end{split}
\end{equation*} 
and hence
\begin{equation*}
\begin{split}
       & \mathrm{det}(sI+A^T)\mathrm{det}(sI-A)\mathrm{det}(\Phi_J(s)) \\
        & = \mathrm{det}(\Theta_{n_u})\mathrm{det}\left(\begin{bmatrix} 
        sI-A & -B_u\Theta_{n_u}B^T_u \\
        -C^T\Theta_{n_y}C & sI + A^T
        \end{bmatrix}\right),
\end{split}
\end{equation*}
where  $\mathrm{det}(\Theta_{n_u})=1$.
Now, using the Hamiltonian matrix~(\ref{hamil}), we get
\begin{align} \label{detphij}
    \mathrm{det}(\Phi_J(s)) = \frac{\mathrm{det}(sI-H)}{\mathrm{det}(sI+A^T)\mathrm{det}(sI-A)}.
\end{align}
Furthermore, using the fact that there exist a skew-symmetric solution $X$ of the NSARE~(\ref{NSARE}), it follows that $\Phi_J(s)$ can be represented as
\begin{equation*}
\begin{split}
   & \Phi_J(s)\\ 
   & = \Theta_{n_u} - B_u^T(-sI-A^T)^{-1}C^T\Theta_{n_y}C~(sI-A)^{-1}B_u \\
   & = \Theta_{n_u} - B_u^T[X(sI-A)^{-1}+(-sI-A^T)^{-1}X \\
   & +(-sI-A^T)^{-1}XB_u\Theta_{n_u}B_u^TX(sI-A)^{-1}]B_u \\
   & = \Theta_{n_u} - B_u^TX(sI-A)^{-1}B_u - B_u^T(-sI-A^T)^{-1}XB_u \\
   & - B_u^T(-sI-A^T)^{-1}XB_u\Theta_{n_u}B_u^TX(sI-A)^{-1}B_u \\
   & = [I + B_u^T(-sI-A^T)^{-1}XB_u\Theta_{n_u}]\Theta_{n_u}\\&\qquad \qquad \qquad \qquad \times [I + \Theta_{n_u}B_u^TX(sI-A)^{-1}B_u].
\end{split} 
\end{equation*}
This implies
\begin{multline}\label{finddetphiJ}
    \mathrm{det}(\Phi_J(s)) = \mathrm{det}(I + (-B_u^T)(sI+A^T)^{-1}XB_u\Theta_{n_u}) \\ 
   \times \mathrm{det}(\Theta_{n_u})\mathrm{det}(I +  \Theta_{n_u}B_u^TX(sI-A)^{-1}B_u).
\end{multline}
Since $\mathrm{det}(\Theta_{n_u})=1$, then
\begin{multline*}
    \mathrm{det}(\Phi_J(s)) = \mathrm{det}(I + (-B_u^T)(sI+A^T)^{-1}XB_u\Theta_{n_u}) \\ 
   \times \mathrm{det}(I + \Theta_{n_u}B_u^TX(sI-A)^{-1}B_u).
\end{multline*}
Now, considering the right-hand side of this equation using the determinant relation from~\cite[p. 135, Fact 2.14.13]{Bernstein}, it follows that 
\begin{equation*}
\begin{split}
       & \mathrm{det}(sI+A^T)\mathrm{det}(I + (-B_u^T)(sI+A^T)^{-1}XB_u\Theta_{n_u}) \\ 
       & = \mathrm{det}(I)\mathrm{det}(sI+A^T-XB_u\Theta_{n_u}I^{-1}B_u^T) \\
       & = \mathrm{det}(sI+A^T-XB_u\Theta_{n_u}B_u^T)
\end{split}
\end{equation*}and hence
\begin{equation}\label{detphijwithM1}
\begin{split}
       & \mathrm{det}(I + (-B_u^T)(sI+A^T)^{-1}XB_u\Theta_{n_u}) \\ 
       & = \frac{\mathrm{det}(sI+A^T-XB_u\Theta_{n_u}B_u^T)}{\mathrm{det}(sI+A^T)}.
\end{split}
\end{equation}
Similarly,
\begin{equation*}
\begin{split}
       & \mathrm{det}(sI-A)\mathrm{det}(I + \Theta_{n_u}B_u^TX(sI-A)^{-1}B_u) \\ 
       & = \mathrm{det}(I)\mathrm{det}(sI-A+B_uI^{-1}\Theta_{n_u}B_u^TX) \\
       & = \mathrm{det}(sI-A+B_u\Theta_{n_u}B_u^TX)
\end{split}
\end{equation*}and hence
\begin{equation}\label{detphijwithM2}
\begin{split}
       & \mathrm{det}(I + \Theta_{n_u}B_u^TX(sI-A)^{-1}B_u)  \\ 
       & = \frac{\mathrm{det}(sI-A+B_u\Theta_{n_u}B_u^TX)}{\mathrm{det}(sI-A)}.
\end{split}
\end{equation}
Substituting equations~(\ref{detphijwithM1}) and~(\ref{detphijwithM2}) into equation~(\ref{finddetphiJ}), it follows that
\begin{align*}
    & \mathrm{det}(\Phi_J(s))\\ & = \frac{\mathrm{det}(sI-A+B_u\Theta_{n_u}B_u^TX)}{\mathrm{det}(sI-A)} 
   \frac{\mathrm{det}(sI+A^T-XB_u\Theta_{n_u}B_u^T)}{\mathrm{det}(sI+A^T)}.
\end{align*}
Substituting $\mathrm{det}(\Phi_J(s))$ using~(\ref{detphij}), it follows that 
\begin{multline}\label{fulldet}
    \frac{\mathrm{det}(sI-H)}{\mathrm{det}(sI+A^T)\mathrm{det}(sI-A)}\\ = \frac{\mathrm{det}(sI-A+B_u\Theta_{n_u}B_u^TX)}{\mathrm{det}(sI-A)} \\
   \times \frac{\mathrm{det}(sI+A^T-XB_u\Theta_{n_u}B_u^T)}{\mathrm{det}(sI+A^T)}.
\end{multline}
Now, using the Matrix Inversion Lemma~\cite{Bernstein}, the inverse of $N(s)$ can be calculated as
\begin{equation}
    N^{-1}(s) = I - \Theta_{n_u}B^T_uX(sI-A+B_u\Theta_{n_u}B^T_uX)^{-1}B_u. 
\end{equation}
This shows that the poles of $N(s)$ and $N^{-1}(s)$ serve as part of the denominator and numerator of equation~(\ref{fulldet}), respectively. 
Looking at Assumption~\ref{assumption3} and considering Remark~\ref{remark}, it follows that there can be no pole-zero cancellation on the left hand side of equation~(\ref{fulldet}).
Therefore, no pole-zero cancellation can exist in the right hand side of equation~(\ref{fulldet}). Thus, condition~\ref{d3} has been shown. 

Now observe that condition~\ref{d4} of Definition~\ref{definition1} follows directly from equation~(\ref{N}). 
Thus, we have demonstrated that $N(s)$ defined in~(\ref{N}) indeed defines a J-spectral factorization of $\Phi_J(s)$.

Sufficiency: Conversely, suppose that $\Phi_J(s)$ has a J-spectral factorization $\Phi_J(s)=N^\sim(s)\Theta_{n_u}N(s)$. 
We will show there exists a skew-symmetric solution $X$ of the NSARE~(\ref{NSARE}). 
Firstly, recall equation~(\ref{detphij}) and consider Assumptions~\ref{assumption1} and~\ref{assumption3}; it follows that there will be no pole zero cancellations in this expression. 
Thus, $\mathrm{det}(\Phi_J(s))$ must be of McMillan degree $2n$.
We now establish some useful claims to aid the proof. 
\begin{claim}
The matrix $N(s)$ is of McMillan degree $n$.
\end{claim}

To establish this claim, first let 
\begin{equation*}
    \mathrm{det}(N(s)) = \frac{\pi_n (s)}{\pi_d(s)}
\end{equation*}
where $\pi_n(s)$ and $\pi_d(s)$ (relatively prime) are of the same degree. Then, from equation~(\ref{detphij}) and condition~\ref{d1} of Definition~\ref{definition1} implies
\begin{align*}
    \mathrm{det}(\Phi_J) & = \frac{\mathrm{det}(sI-H)}{\mathrm{det}(sI+A^T)\mathrm{det}(sI-A)} \\
    & = \frac{\pi_{n^\sim}(s)\pi_n(s)}{\pi_{d^\sim}(s)\pi_d(s)}
\end{align*} with a McMillan degree $2n$.
From~\ref{d3} of Definition~\ref{definition1}, we obtain
\begin{align*}
    & \pi_{n^\sim}(s)\pi_n(s) = \mathrm{det}(sI-H) \\
    & \pi_{d^\sim}(s)\pi_d(s) = \mathrm{det}(sI+A^T)\mathrm{det}(sI-A).
\end{align*}
We can write 
\begin{align*}
    & \pi_{d^\sim}(s) = n_a, \\
    & \pi_d(s) = n_b.
\end{align*}
The McMillan degrees are related by \cite{Kalman}
\begin{equation*}
    2n \leq n_a + n_b.
\end{equation*}
Hence, $n_a = n_b = n$ and $\pi_d(s)$ will be of degree of $n$ and thus $N(s)$ will have McMillan degree $n$. This completes the proof of the claim.

We now define the matrix $L$ to be a skew-symmetric solution to the Lyapunov equation
\begin{equation}\label{lyapL}
    A^TL+LA+C^T\Theta_{n_y}C = 0.
\end{equation}

\begin{claim}\label{claim2}
The matrix $\Phi_J$ can be written in the form
\begin{equation*}
    \Phi_J(s)= \Theta_{n_u} - B_u^T(-sI-A^T)^{-1}LB_u-B_u^TL(sI-A)^{-1}B_u
\end{equation*} 
such that the realization of $-B_u^T(-sI-A^T)^{-1}LB_u$ and $-B_u^TL(sI-A)^{-1}B_u$ are minimal. Hence, the pair $(A, B_u)$ is controllable.
\end{claim}

To establish this claim, first observe that~(\ref{lyapL}) implies 
\begin{align*}
    & L(sI-A)^{-1}-(-sI-A^T)^{-1}L \\ & =(-sI-A^T)^{-1}C^T\Theta_{n_y}C(sI-A)^{-1}.
\end{align*}
Hence, we can write
\begin{align}
     \Phi_J(s) & = \Theta_{n_u} - B_u^T(-sI-A^T)^{-1}C^T\Theta_{n_u}C~(sI-A)^{-1}B_u \nonumber \\
     & = \Theta_{n_u} - B_u^T[L(sI-A)^{-1}-(-sI-A^T)^{-1}L]B_u \nonumber \\
     & = \Theta_{n_u} - B_u^T(-sI-A^T)^{-1}LB_u-B_u^TL(sI-A)^{-1}B_u \label{equality}
\end{align} 
as required. 
Also, it follows from this expression for $\Phi_J(s)$ that its McMillan degree satisfies the inequality
\begin{align*}
    \delta(\Phi_J(s)) & \leq \delta(- B_u^T(-sI-A^T)^{-1}LB_u) \\ 
    & + \delta(-B_u^TL(sI-A)^{-1}B_u) \\
    & \leq n + n.
\end{align*}
We know that $\delta(\Phi_J(s)) = 2n$ and hence equality~(\ref{equality}) must hold. 
Therefore, $\delta(- B_u^T(-sI-A^T)^{-1}LB_u) = n$ and $\delta(-B_u^TL(sI-A)^{-1}B_u) = n$. 
That is the realizations $- B_u^T(-sI-A^T)^{-1}LB_u$ and $-B_u^TL(sI-A)^{-1}B_u$ are minimal.

It remains to show that $L$ is a skew-symmetric solution to~(\ref{lyapL}). Add equations~(\ref{lyapL}) and the transpose of~(\ref{lyapL})
\begin{align*}
     A^TL+LA+C^T\Theta_{n_y}C + A^TL^T+L^TA-C^T\Theta_{n_y}C & = 0 \\
     A^T(L+L^T) + (L+L^T)A & = 0.
\end{align*}
Since $A$ is Hurwitz according to Assumption \ref{assumption1}, this implies that $L+L^T = 0$ i.e., $L^T = -L$. This completes the proof of the claim. 

\begin{claim}\label{claim3}
There exists a minimal realization of $N(s)$ of the form $N(s)=I+N_A(sI-A)^{-1}B_u$.
\end{claim}

To establish this claim, first recall condition~\ref{d4} of Definition~\ref{definition1}. 
Hence, $N(s)$ will have a minimal realization of the form $N(s)=I+N_A(sI-F)^{-1}E$ where $F$ is a $n \times n$ matrix. 
Also, since we know $N(s)$ is analytic in Re $s \geq 0$ therefore the matrix $F$ will be Hurwitz. 
We can now define the matrix $Y$ to be the solution to the Lyapunov equation
\begin{equation}\label{lyapY}
    F^TY + YF + N_A^T\Theta_{n_u}N_A = 0.
\end{equation}
Using this equation, it follows that 
\begin{align*}
    & Y(sI-F)^{-1}+(-sI-F^T)^{-1}Y \\ & =(-sI-F^T)^{-1}N_A^T\Theta_{n_u}N_A(sI-F)^{-1}.
\end{align*}
Thus, 
\begin{equation}\label{eqn1}
    \begin{split}
        \Phi_J(s) 
     & = N^\sim(s)\Theta_{n_u}N(s) \\ 
     & = (I+N_A(sI-F)^{-1}E)^\sim \Theta_{n_u}(I+N_A(sI-F)^{-1}E) \\ 
     & = \Theta_{n_u}+\Theta_{n_u}N_A(sI-F)^{-1}E+E^T(-sI-F^T)^{-1}N_A^T\Theta_{n_u} \\
     & + E^T(-sI-F^T)^{-1}N_A^T\Theta_{n_u}N_A(sI-F)^{-1}E \\
     & = \Theta_{n_u}+\Theta_{n_u}N_A(sI-F)^{-1}E+E^T(-sI-F^T)^{-1}N_A^T\Theta_{n_u} \\
     & + E^T[Y(sI-F)^{-1}+(-sI-F^T)^{-1}Y]E \\
     & = \Theta_{n_u}+\Theta_{n_u}N_A(sI-F)^{-1}E+E^T(-sI-F^T)^{-1}N_A^T\Theta_{n_u} \\
     & + E^TY(sI-F)^{-1}E+E^T(-sI-F^T)^{-1}YE \\
     & = \Theta_{n_u}+(\Theta_{n_u}N_A+E^TY)(sI-F)^{-1}E \\
     & +E^T(-sI-F^T)^{-1}(N_A^T\Theta_{n_u}+YE)
    \end{split}
\end{equation}
Using this equation, it follows that the McMillan degree of $F(s)$ satisfies the inequality
\begin{align*}
    \delta(\Phi_J(s)) & \leq \delta((\Theta_{n_u}N_A+E^TY)(sI-F)^{-1}E) \\ 
    & + \delta(E^T(-sI-F^T)^{-1}(N_A^T\Theta_{n_u}+YE)) \\
    & \leq n + n.
\end{align*}
However, we know that $\delta(\Phi_J(s)) = 2n$ and hence equality must hold. Thus, $\delta((\Theta_{n_u}N_A+E^TY)(sI-F)^{-1}E)=n$ and $\delta(E^T(-sI-F^T)^{-1}(N_A^T\Theta_{n_u}+YE))=n$. That is, the realizations $(\Theta_{n_u}N_A+E^TY)(sI-F)^{-1}E)$ and $E^T(-sI-F^T)^{-1}(N_A^T\Theta_{n_u}+YE)$ are minimal.

We now compare the expression for $\Phi_J(s)$ obtained in Claim~\ref{claim2} and equation~(\ref{eqn1}). It follows that
\begin{align*}
    \Phi_J(s) & = \Theta_{n_u} - B_u^T(-sI-A^T)^{-1}LB_u-B_u^TL(sI-A)^{-1}B_u \\
    & = \Theta_{n_u}+(\Theta_{n_u}N_A+E^TY)(sI-F)^{-1}E \\
    & + E^T(-sI-F^T)^{-1}(N_A^T\Theta_{n_u}+YE). 
\end{align*}
Equating stable and anti-stable terms in this equation, it follows that
\begin{equation}\label{eqn2}
    B_u^T(-sI-A^T)^{-1}(-LB_u) = E^T(-sI-F^T)^{-1}(N_A^T\Theta_{n_u}+YE)
\end{equation}and 
\begin{equation}\label{eqn3}
    -B_u^TL(sI-A)^{-1}B_u = (\Theta_{n_u}N_A+E^TY)(sI-F)^{-1}E.
\end{equation}
However, both sides of equation~(\ref{eqn3}) are minimal realizations. 
Thus, we conclude that there exists a minimal realization of $N(s)$ such that $F=A$ and $E=B_u$. 
In a similar manner to Claim~\ref{claim2}, it remains to show that $Y$ is a skew-symmetric solution to~(\ref{lyapY}). 
Adding equations~(\ref{lyapY}) and its transpose
\begin{align*}
     F^TY+YF+N_A^T\Theta_{n_u}N_A + F^TY^T+Y^TF-N_A^T\Theta_{n_u}N_A & = 0 \\
     F^T(Y+Y^T) + (Y+Y^T)F & = 0.
\end{align*}
Since $F$ is Hurwitz, this implies that $Y+Y^T = 0$ i.e., $Y^T = -Y$. This completes the proof of the claim. 

Now, returning to the proof of the theorem, substitute $F=A$ and $E=B_u$ into equations~(\ref{eqn2}) and~(\ref{eqn3}). Thus
\begin{equation*}
    B_u^T(-sI-A^T)^{-1}(-LB_u) = B_u^T(-sI-A^T)^{-1}(N_A^T\Theta_{n_u}+YB_u)
\end{equation*}and 
\begin{equation*}
   -B_u^TL(sI-A)^{-1}B_u = (\Theta_{n_u}N_A+B_u^TY)(sI-A)^{-1}B_u.
\end{equation*}
However, we know from Claim~\ref{claim2} that equations~(\ref{eqn2}) and~(\ref{eqn3}) are minimal realizations. Therefore, it follows from the result above that
\begin{align*}
    -LB_u & = N_A^T\Theta_{n_u}+YB_u \\
    -B_u^TL & = \Theta_{n_u}N_A+B_u^TY.
\end{align*}
Defining $X=L+Y$, then
\begin{align*}
   -XB_u & = N_A^T\Theta_{n_u} \\
    -B_u^TX & = \Theta_{n_u}N_A,
\end{align*}
which provides the desired realizations. 
It now remains to show that $X$ is a solution of the NSARE~(\ref{NSARE}). Add~(\ref{lyapL}) and~(\ref{lyapY}), this gives
\begin{equation*}
    A^TL+LA+C^T\Theta_{n_y}C+F^TY + YF + N_A^T\Theta_{n_u}N_A=0.
\end{equation*}
Hence, substitute $X=L+Y$ and the result above gives
\begin{equation*}
    A^TX+XA+C^T\Theta_{n_y}C-XB_u\Theta_{n_u}B_u^TX=0
\end{equation*} as required.
In addition, from Claim~\ref{claim2} and~\ref{claim3}, we know that $L$ and $Y$ are skew-symmetric solution. 
Therefore, $X$ is a skew-symmetric solution of the NSARE~(\ref{NSARE}). 
It now follows from Theorem~\ref{thm2} that $G(s)$ is physically realizable with only direct feedthrough quantum noise. 
\end{proof}

\subsection{Frequency Response Condition}
The following corollary gives a necessary frequency response condition for physical realizability with only direct feedthrough quantum noise.
\begin{corollary}\label{corollary}
Suppose $G(s)$ is physically realizable with only direct feedthrough quantum noise.
Then the following frequency response condition 
\begin{equation}\label{freqdomcond}
    \mathrm{det}(\Theta_{n_u}-G^{\dagger}(j\omega)\Theta_{n_y}G(j\omega)) > 0
\end{equation}
holds for all $\omega$.
\end{corollary}

\begin{proof}
From Theorem~\ref{theorem}, it follows that $\Phi_J(s)$ has a J-spectral factorization. 
$\Phi_J(s)$ can also be written as 
    \begin{equation*}
    \Phi_J(s) = \Theta_{n_u}-G^{\sim}(s)\Theta_{n_y}G(s).
    \end{equation*}
Now, we let $s=j\omega$ 
    \begin{equation}\label{phijomega}
    \Phi_J(j\omega) = \Theta_{n_u}-G^{\dagger}(j\omega)\Theta_{n_y}G(j\omega)
    \end{equation}
and take the determinant of equation~(\ref{phijomega})
    \begin{equation}\label{phijomegadet}
    \mathrm{det}(\Phi_J(j\omega)) = \mathrm{det}(\Theta_{n_u}-G^{\dagger}(j\omega)\Theta_{n_y}G(j\omega)).
    \end{equation}
First, look at the right-hand side of this expression as $\omega \rightarrow \infty$, and note that given $G(s)$ is strictly proper
    \begin{equation*}
        \mathrm{det}(G^{\dagger}(j\omega)\Theta_{n_y}G(j\omega)) \rightarrow 0
    \end{equation*}
as $\omega \rightarrow \infty$. Hence, 
    \begin{equation*}
        \mathrm{det}(\Phi_J(j\omega)) \rightarrow \mathrm{det}(\Theta_{n_u}) = 1 > 0
    \end{equation*}
as $\omega \rightarrow \infty$.

Also, it follows from~\ref{d1} in Definition~\ref{definition1}
that\begin{equation*}
        \mathrm{det}(\Phi_J(j\omega)) = \mathrm{det}(N^{\dagger}(j\omega)\Theta_{n_u}N(j\omega))
\end{equation*}
for all $\omega$. 

Furthermore, the matrix $iN^{\dagger}(j\omega)\Theta_{n_u}N(j\omega)$ is congruent to the matrix $i\Theta_{n_u}$ which has $\frac{n_u}{2}$ positive eigenvalues and $\frac{n_u}{2}$ negative eigenvalues. 
Hence, using the Inertia Theorem~\cite[pp. 281-282, Definition 4.5.6]{HornJohnson}, it follows that the eigenvalues of $iN^{\dagger}(j\omega)\Theta_{n_u}N(j\omega)$ will have this same property. 
From this, it follows that the $\det(N^{\dagger}(j\omega)\Theta_{n_u}N(j\omega))$ must be real and non-zero for any $\omega$. Hence, it now follows that
 \begin{equation*}
        \mathrm{det}(\Phi_J(j\omega)) > 0
    \end{equation*} for all $\omega$.
\end{proof}

%%%%%%%%%%%%%%%%%%%%%%%%%%%%%%%%%%%%%%%%%%%%%%%%%%%%%%%%%%%%%%%%%%
\section{Examples}\label{example}

\subsection{Example 1}
We first demonstrate our results through an example which considers a system from~\cite{JamesNurdinPetersen}. 
Given matrices $A$, $B_u$ and $C$ and corresponding transfer function matrix $G(s) = C(sI-A)^{-1}B_u$. 
Here $A$ is Hurwitz and has no common eigenvalues with the Hamiltonian $H$:
\begin{equation*}
            \begin{split}
               A & =\begin{bmatrix} 
                        -1.3894\times I & -0.4472\times I \\
                        -0.2\times I & -0.25\times I
                    \end{bmatrix}; \\
                B_u & =\begin{bmatrix} 
                        -0.4472\times I\\ 0_{2 \times 2}
                    \end{bmatrix}; \\
              C & =\begin{bmatrix} 
                        -0.4472\times I & 0_{2 \times 2}.
                    \end{bmatrix}
            \end{split}
\end{equation*} 

\subsubsection{J-Spectral Factorization}
First, form $G(s)=C(sI-A)^{-1}B_u$ to obtain
\begin{equation*}
\begin{aligned}
 G(s) &=
\left[\begin{matrix}
            \frac{0.20s^3+0.38s^2+0.133s+0.01}{s^4+3.279s^3+3.203s^2+0.8456s+0.07} \\
            0
\end{matrix}\right.\\
&\qquad\qquad
\left.\begin{matrix}
           0 \\
            \frac{0.20s^3+0.38s^2+0.133s+0.01}{s^4+3.279s^3+3.203s^2+0.8456s+0.07} 
\end{matrix}\right]
\end{aligned}
\end{equation*}
and construct $\Phi_J(s)=\Theta_u-G^\sim(s)\Theta_{n_y}G(s)$ to obtain
\begin{equation}\label{exphij}
\begin{aligned}
\Phi_J(s) &=
\left[\begin{matrix}
            0 \\
            \frac{-s^4-3.28s^3-3.16s^2-0.83s-0.06}{s^4+3.28s^3+3.20s^2+0.85s+0.07}
\end{matrix}\right.\\
&\qquad\qquad
\left.\begin{matrix}
           \frac{s^4+3.28s^3+3.16s^2+0.826s+0.06}{s^4+3.28s^3+3.20s^2+0.85s+0.07} \\
           0
\end{matrix}\right].
\end{aligned}
\end{equation}

After a process of trial and error, we find a suitable $N(s)$ as
\begin{equation*}
N(s) =
\begin{bmatrix}
            \frac{s^2+1.62s+0.25}{s^2+1.64s+0.26} & 0\\
            0 & \frac{s^2+1.62s+0.25}{s^2+1.64s+0.26}
\end{bmatrix}.
\end{equation*}
It is straightforward to verify that this $N(s)$ satisfies the conditions in Definition~\ref{definition1}. 
Now using the method described in~\cite{VuglarPetersen}, the calculated solution $X$ of~(\ref{NSARE}) is
\begin{equation*}
X =
\begin{bmatrix}
    -0.0000   & 0.0763  &  0.0000 &  -0.0270 \\
    -0.0763   & -0.0000  &  0.0270  & -0.0000 \\ 
    -0.0000   & -0.0270 &   0.0000 &   0.0486 \\
    0.0270    & 0.0000 &  -0.0486  &  0.0000
\end{bmatrix}.
\end{equation*} 
It is straightforward to verify that this value of $X$ is real, non-singular and skew-symmetric.
Next, we check~(\ref{NSARE}) and find
\begin{equation*}
         A^TX + XA - XB_u\Theta_{n_u}B_u^TX + C^T\Theta_{n_y}C = 0.
    \end{equation*}
Thus, the calculated $X$ above solves the NSARE~(\ref{NSARE}). 
From Theorem~\ref{theorem}, we conclude that $G(s)$ is physically realizable using only direct feedthrough quantum noise where we calculate the matrix $B_v$ using the approach of \cite{VuglarPetersen} as
\begin{align*}
    B_{v} & = \Theta C^T diag(J) \\
    & = \begin{bmatrix} 
         0.4472    & 0\\
         0    & 0.4472\\
         0         & 0\\
         0         & 0
        \end{bmatrix}.
\end{align*}
This formula comes from Theorem~\ref{theoremPR}, \cite{VuglarPetersen} uses a state space transformation so that $X$ $\rightarrow$ $\Theta_n$. The matrix $C$ would need to be transformed via the same state space transformation.
\subsubsection{Frequency Response Condition}
In the above example, we showed that $G(s)$ is physically realizable with only direct feedthrough quantum noise. 
Therefore, according to Corollary~\ref{corollary}, the frequency response condition~(\ref{freqdomcond}) should hold. 
This can be seen in Figure~\ref{fig:system}.
 \begin{figure}[htp!] 
    \centering
    \includegraphics[width=8cm]{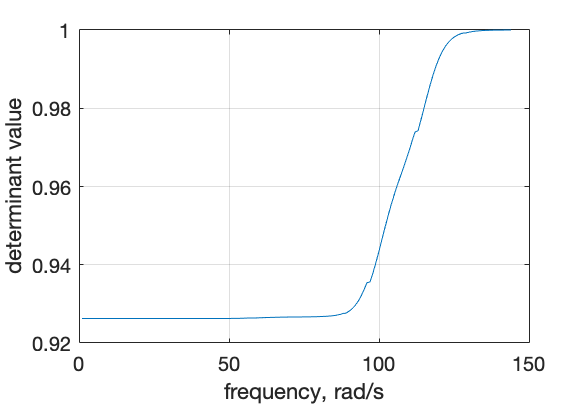}
    \caption{Plot of~$\mathrm{det}(\Theta_{n_u}-G^{\sim}(j\omega)\Theta_{n_u}G(j\omega))$ versus $\omega$.}
    \label{fig:system}
\end{figure}

\subsection{Example 2}
We now give an example where the frequency response condition~(\ref{freqdomcond}) holds but the condition for a transfer function matrix to be physically realizable in Theorem~\ref{thm2} fails. Consider the matrices 
\begin{equation*}
            \begin{split}
               A & =\begin{bmatrix} 
                            0.8844  &  0.4385 &   0.3249   & 0.5466 \\
                            0.7209  &  0.4378 &   0.2462 &   0.5619\\
                            0.0186  &  0.1170  &  0.3427  &  0.3958\\
                            0.6748  &  0.8147  &  0.3757  &  0.3981
                    \end{bmatrix}; \\
                B_u & =\begin{bmatrix} 
                        0.5154 &   0.4001\\
                        0.6575  &  0.8319\\
                        0.9509  &  0.1343\\
                        0.7223  &  0.0605
                    \end{bmatrix}; \\
              C & =\begin{bmatrix} 
                        0.0842   &  0.3242  &  0.0117  &  0.0954\\
                        0.1639  &  0.3017 &   0.5399  &  0.1465
                    \end{bmatrix}.
            \end{split}
\end{equation*} 
Figure~\ref{fig:system2} gives a plot of the corresponding frequency response condition~(\ref{freqdomcond}).
\begin{figure}[htp!] 
    \centering
    \includegraphics[width=8cm]{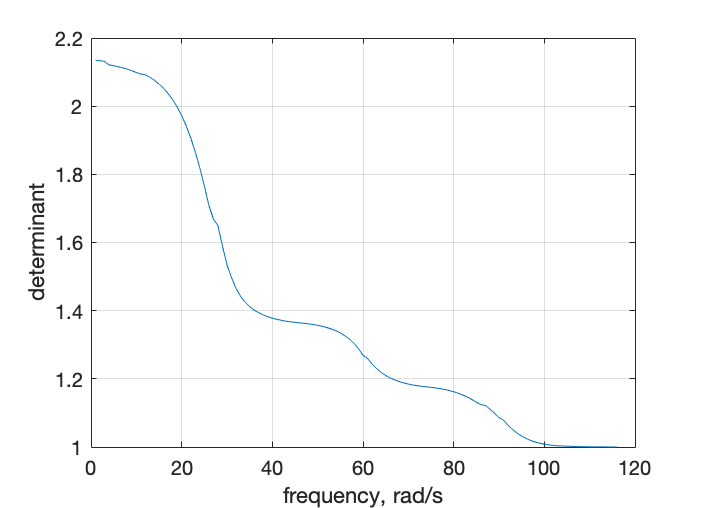}
    \caption{Plot of~$\mathrm{det}(\Theta_{n_u}-G^{\sim}(j\omega)\Theta_{n_u}G(j\omega))$ versus $\omega$.}
    \label{fig:system2}
\end{figure}
It can be seen in Figure~\ref{fig:system2} that the frequency response condition~(\ref{freqdomcond}) holds. 
Now following the steps in~\cite{VuglarPetersen} to find a solution for the Riccati equation~(\ref{NSARE}), we calculate the eigenvalues of the Hamiltonian matrix $H$~(\ref{hamil}) and find that this matrix has no purely imaginary eigenvalues. 
We form the matrix $\begin{bmatrix} X_1 \\ X_2 \end{bmatrix}$ of eigenvectors corresponding to the eigenvalues, $\lambda$ of $H$ satisfying $Re(\lambda) < 0$. 
In particular, the matrix $X_1$ is found as  
\begin{equation*}
\begin{aligned}
X_1 & = 
\left[\begin{matrix}
            0.0731 + 0.0000i  & 0.1216 + 0.0000i \\
            0.2869 + 0.0000i &  0.5130 + 0.0000i \\
            -0.1730 + 0.0000i  & 0.3364 + 0.0000i \\
            -0.2628 + 0.0000i & -0.6759 + 0.0000i
\end{matrix}\right.\\
&\qquad\qquad
\left.\begin{matrix}
           & 0.1166 - 0.0240i &  0.1166 + 0.0240i \\
           & 0.1546 - 0.0660i &  0.1546 + 0.0660i \\
           & 0.2213 + 0.2812i  & 0.2213 - 0.2812i \\
           & -0.5983 + 0.0000i & -0.5983 + 0.0000i
\end{matrix}\right]. 
\end{aligned}
\end{equation*} 
It is straightforward to verify that this matrix $X_1$ is singular. 
Hence, it follows as in~\cite{VuglarPetersen} that the Riccati equation~(\ref{NSARE}) does not have a solution.
Therefore, it follows from Theorem~\ref{thm2} that this $G(s)=C(sI-A)^{-1}B_u$ is not physically realizable with only direct feedthrough quantum noise.
Thus, we have shown that the frequency response condition~(\ref{freqdomcond}) is not sufficient for physical realizability with only direct feedthrough quantum noise as defined in Definition \ref{dfn2}.
%%%%%%%%%%%%%%%%%%%%%%%%%%%%%%%%%%%%%%%%%%%%%%%%%%%%%%%%%%%%%%%%%%%%%%%%%%%%%%%%
\section{CONCLUSION AND FUTURE WORK}\label{conclusion}

\subsection{Conclusion}
In~\cite{VuglarPetersen}, the existence of a non-singular, real, and skew-symmetric solution $X$ to the NSARE~(\ref{NSARE}) guarantees a strictly proper transfer function can be implemented as a physically realizable quantum system with only direct feedthrough noise. In this work, we have shown that the J-spectral factorization of $\Phi_J(s)$ gives a necessary and sufficient condition for $G(s)$ to be physically realizable with only direct feedthrough quantum noise. We also present a necessary frequency response condition. 

\subsection{Future Work}
Future work will consider making the frequency response condition in Corollary~\ref{corollary} to a necessary and sufficient condition by somehow extending the notion of physically realizability with only direct feedthrough noise.
%%%%%%%%%%%%%%%%%%%%%%%%%%%%%%%%%%%%%%%%%%%%%%%%%%%%%%%%%%%%%%%%%%%%%%%%%%%%%%%%


\begin{thebibliography}{99}

\bibitem{Molinari}
B. P. Molinari, Equivalence Relations for the Algebraic \text{R}iccati Equation, {\it The Bell System Technical Journal}, vol. 42, 1963, pp 355-382.

\bibitem{Kalman}
R. E. Kalman, Irreducible Realizations and the Degree of a Rational Matrix, {\it Journal of the Society for Industrial and Applied Mathematics}, vol. 13, 1965, pp 520-544.

\bibitem{HudsonParthasarathy}
R. L. Hudson and K.R. Parthasarathy, Quantum \text{I}to's formula and stochastic evolutions, {\it Communications in Mathematical Physics}, vol. 93, 1984, pp 301-323.

\bibitem{Belavkin}
V. P. Belavkin, Quantum continual measurements and a posteriori collapse on \text{CCR}, {\it Communications in Mathematical Physics}, vol. 146, 1992, pp 611-631.

\bibitem{Parthasarathy}
K.R. Parthasarathy, {\it An Introduction to Quantum Stochastic Calculus}, Birkhäuser Basel; 1992.

\bibitem{ZhouDoyleGlover}
K. Zhou and J. C. Doyle and K. Glover, {\it Robust and Optimal Control}, Prentice Hall Press, New Jersey; 1996.

\bibitem{Kimura}
H. Kimura, {\it Chain-scattering approach to H$^{\infty}$ control}, Birkhäuser; 1997.

\bibitem{GardinerZoller}
C. Gardiner and P. Zoller, {\it Quantum Noise}, Springer-Verlag Berlin Heidelberg; 2004.

\bibitem{BachorRalph}
Hans‐A. Bachor and T. C. Ralph, {\it A Guide to Experiments in Quantum Optics}, Wiley-VCH; 2004.

\bibitem{WallsMilburn}
D.F. Walls and G. J. Milburn, {\it Quantum Optics}, Springer-Verlag Berlin Heidelberg; 2008.

\bibitem{JamesNurdinPetersen}
M. R. James and H. I. Nurdin and I. R. Petersen, H$^{\infty}$ Control of Linear Quantum Stochastic Systems, {\it IEEE Transactions on Automatic Control}, vol. 53, 2008, pp 1787-1803.

\bibitem{Bernstein}
D. Bernstein, {\it Matrix \text{M}athematics: Theory, Facts, and Formulas: (Second edition)}, Princeton University Press; 2009.

\bibitem{HornJohnson}
R. A. Horn and C. R. Johnson, {\it Matrix Analysis}, Cambridge University Press, New York; 2013.

\bibitem{VuglarPetersen}
Shanon L. Vuglar and Ian R. Petersen, Quantum Noises, Physical Realizability and Coherent Quantum Feedback Control, {\it IEEE Transactions on Automatic Control}, vol. 62, 2017, pp 998-1003.

\bibitem{Stefanovski}
J. D. Stefanovki, Strongly (J,J') lossless rational matrices and H$^{\infty}$ problem, {\it International Journal of Robust and Nonlinear Control}, vol. 28, 2018, pp 4261-4286.

\bibitem{ThienVuglarPetersen}
Rebbecca T. TY and Shanon L. Vuglar and Ian R. Petersen, "When do additional quantum noises affect controller performance?", {\it in 2020 59th IEEE Conference on Decision and Control (CDC)}, 2020, pp. 3855-3859.

\end{thebibliography}
\end{document}